\documentclass[12pt]{article}

\usepackage{epsfig}

\usepackage{amssymb}
\usepackage{amsfonts}

\usepackage{color}
 
\def\be{\begin{equation}}
\def\ee{\end{equation}}
\def\ba{\begin{array}{c}}
\def\ea{\end{array}}

\def\ben{$$}
\def\een{$$}

\newcommand{\bea}{\begin{eqnarray}}
\newcommand{\eea}{\end{eqnarray}}

\newcommand{\N}{\mathbb{N}}

\newcommand{\kt}{\rangle}
\newcommand{\br}{\langle}

\newtheorem{thm}{Theorem}

\newtheorem{lemma}[thm]{Lemma}

\newenvironment{proof}{\noindent
 {\bf Proof.}}{\hfill$\square$\vspace{3mm}\endtrivlist}

\begin{document}


\vspace{.35cm}

\begin{center}

{\Large

Passage through exceptional point: Case study

 }

\vspace{10mm}

\textbf{Miloslav Znojil}

The Czech Academy of Sciences, Nuclear Physics Institute,

 Hlavn\'{\i} 130,
250 68 \v{R}e\v{z}, Czech Republic

\vspace{0.2cm}

 and

\vspace{0.2cm}

Department of Physics, Faculty of Science, University of Hradec
Kr\'{a}lov\'{e},

Rokitansk\'{e}ho 62, 50003 Hradec Kr\'{a}lov\'{e},
 Czech Republic

%
%
%
\end{center}



\section*{Abstract}

The description of unitary evolution using non-Hermitian but
``hermitizable'' Hamiltonians $H$ is feasible via an {\it ad hoc}
metric $\Theta=\Theta(H)$ and a (non-unique) amendment $\br
\psi_1|\psi_2\kt \to \br \psi_1|\Theta |\psi_2\kt $ of the inner
product in Hilbert space. Via a proper fine-tuning of $\Theta(H)$
this opens the possibility of reaching the boundaries of stability
(i.e., exceptional points, EPs) in many quantum systems sampled here
by the fairly realistic Bose-Hubbard (BH) and discrete anharmonic
oscillator (AO) models. In such a setting it is conjectured that the
EP singularity can play the role of a quantum phase-transition
interface between different dynamical regimes. Three alternative
``AO $\leftrightarrow$ BH`` implementations of such an EP-mediated
dynamical transmutation scenario are proposed and shown, at an
arbitrary finite Hilbert-space dimension $N$, exact and
non-numerical.

\subsection*{Keywords}

quasi-Hermitian quantum Hamiltonians; non-Hermitian degeneracies;
transition matrices; phase transitions; closed-form toy models;

\newpage

\section{Introduction}

The recent growth of interest in non-Hermitian and, in particular,
parity-times-time-reversal-symmetric (${\cal PT}-$symmetric) quantum
Hamiltonians $H$ with real spectra \cite{BB,book,Carlbook} upgraded
also the status of Kato's exceptional points (EPs, \cite{Kato}) from
an abstract, purely mathematical concept to an interesting,
experimentally accessible singularity
\cite{Nimrod,stellenbosch,bosch,boscha,boschb,boschc,osch,oschb}. As
a a sample of this tendency one can recall paper~\cite{Uwe}, the
authors of which managed to throw new light upon the traditional
phenomenological problem of Bose-Einstein condensation as well as
upon its EP {\it alias\,} ``non-Hermitian degeneracy'' \cite{Berry}
mathematical background. For the sake of definiteness the authors of
paper~\cite{Uwe} choose the popular Bose-Hubbard (BH) Hamiltonian in
its specific non-Hermitian rearrangement. Their semi-analytic
description of the EP-related phenomena proved relevant for an
improvement of our understanding of many other EP-related models
\cite{Berryb,Berryba,Berrybb,Berrybc}. At present one can observe
that this direction of research did already branch into a large
number of various separate experimental as well as theoretical
subdirections \cite{Berryc}-\cite{Berrycl}.

In our present paper we intend to give the concept of non-Hermitian
degeneracy another, innovated interpretation of a unitary-evolution
gate, or interface, or transition path connecting two different
evolution phases of a single quantum system. In this perspective the
scope of the pioneering study~\cite{Uwe} was broader, with emphasis
put upon the phenomena of instability and, in general, on the
analysis of the behavior of the so called open, i.e., resonant and
non-unitary quantum systems (and, in particular, of the almost
degenerate BH states) under small perturbations. Interested readers
may find more details also in the recent extension of such a
rigorous open-system analysis in \cite{Viola}. In the narrower
context of the closed, strictly unitarily evolving quantum systems
of our present interest a complementary reading containing a deeper
discussion of physics as well as a few further mathematics-oriented
references may be also found in \cite{corridors,corridorsa}.

The description of our present results will start, in section
\ref{formulapro}, by a sketchy outline of an elementary
classification of quantum phase transitions. We explain there that
for any pair of phenomenological Hamiltonians (i.e., for
$H^{(-)}(t)$ defined at time $t<0$, and for $H^{(+)}(t)$ defined at
time $t>0$) the {\em simultaneous\,} existence of their EP-related
interfaces at $t=0$ can be used for a ``gluing'' or ``matching'' of
the respective evolution phases. We will argue that the idea of the
matching of the phases at a suitable interface proves perfectly
phenomenologically consistent, and that the EP-related fine-tuning
of the parameters appears feasible in multiple models. Their non-BH
versions will be sampled here by a discrete anharmonic-oscillator
quantum system (AO, \cite{maximal}).

From the point of view of physics the EP-related instantaneous $t=0$
loss of the observability ``erases'' all of the traces of the
information carried by the physics-determining parameters at $t <0$.
In section \ref{sushimi} we will add that our requirement of the
continuity of the Hamiltonian at $t=0$ still admits a freedom in the
construction. The idea will be shown to lead, in its application to
the exactly solvable BH and AO toy models, to as many as six
non-equivalent realizations of the process of transmutation (see
Table \ref{tataj} below).

The technical background of our $BH \leftrightarrow AO$ models of
evolution will be described in detail. The $N$ by $N$ transition
matrices will be constructed, in an explicit non-numerical form, in
the BH case (section \ref{jofob}) as well as in the AO case (section
\ref{jofo}). In section \ref{sec6} the closed-form transition matrix
will be also found for the less straightforward setup in which the
model-independent, universal Jordan-block-matching matrices will be
replaced, at $t=0$, by their less usual and model-dependent (though
only marginally more complicated) alternatives.

In the last two sections~\ref{summary} and \ref{sumsummary} we will
summarize our message and emphasize some of its less usual physical
aspects and conceptual consequences. {\it Pars pro toto\,} we will
mention that the above-mentioned erasure of the memory at $t=0$
might be reinterpreted as an ambiguity of the $t>0$ unitary
evolutions.

\section{EP-related models of quantum phase transitions\label{formulapro}}

The existence of the processes of the EP-related degeneracy is, of
course, not restricted to any specific model. The inclusion of any
non-BH- or non-AO-related Hamiltonian will still open the same
natural problem of what happens {\em after\,} any unitary quantum
system in question had {\em crossed\,} its EP singularity. This
problem is also the main challenge addressed in our present paper.

\subsection{Interfaces ($t=0$)\label{tretjak}}

In the vast literature devoted to the EP-related quantum phase
transitions (only partially reviewed, say, in \cite{denis}) the
Hamiltonian is often chosen as non-Hermitian but hermitizable {\it
alias\,} quasi-Hermitian (interested readers may find an exhaustive
explanation of this concept, e.g., in review paper \cite{Geyer}).
Typically, such a Hamiltonian is defined as split in two parts,
 \be
 H(t)=
 \left \{
 \ba
 H^{(-)}(t)\,,\ \ t<0\,,\\
 H^{(+)}(t)\,,\ \ t>0\,.
 \ea
 \right .
 \label{schem}
 \ee
Attention is paid here just to the two basic scenarios. In the
simpler one people study, exclusively, just the process which
precedes the transition. In the EP context, in particular, the
evolution {\em before\,} the collapse is described, at $t<0$, by a
non-Hermitian Hamiltonian $H^{(-)}(t)$ which is hermitizable. In
this sense, the evolution generated by the Hamiltonian is still
standard and unitary.

After the degeneracy, at least a part of the spectrum of the initial
Hamiltonian $H^{(-)}(t)$ may be found complex, being declared less
interesting \cite{BB}. Such a picture of reality found a number of
explicit realizations recently, partly also in the so called
quasi-Hermitian \cite{Geyer} {\it alias\,} ${\cal PT}-$symmetric
\cite{Carl} {\it alias\,} pseudo-Hermitian \cite{ali} quantum
mechanics.

In this framework, typically, the ${\cal PT}-$symmetry of a toy
model $H^{(-)}(t)$ gets spontaneously broken at $t=0$. A $t>0$
completion of the picture is then very often skipped since it would
require the construction and use of some prohibitively complicated
partner Hamiltonian $H^{(+)}(t)$. This can be called an
``unforseeable future'' scenario. In a classification scheme of
Ref.~\cite{denis2} it was given the name of quantum phase transition
of the first kind. The control of the evolution after $t=0$ is given
up there.

Irrespectively of the details of the post-EP dynamics, the
first-kind models are usually characterized by the discontinuity at
$t=t^{(EP)}=0$,
 \be
 \lim_{t \to 0^-}H^{(-)}_{[I]}(t)\ \neq \
 \lim_{t \to 0^+}H^{(+)}_{[I]}(t)\,.
 \label{nematchi}
 \ee
This motivated the authors of Ref.~\cite{denis3} to redirect
attention to the so called quantum phase transitions of the second
kind, with the latter name being reserved to the opposite extreme
with
 \be
 H^{(-)}_{[II]}(t) = H^{(+)}_{[II]}(t)\,
 \ee
at all of the relevant times $t$. In other words, in the
second-kind-transition models the Hamiltonian remains the same all
the time. The passage through the EP singularity is described by the
same Schr\"{o}dinger equation. What is changed at $t=t^{(EP)}=0$ is
not the Hamiltonian itself but merely the inner product in the
physical Hilbert space of states \cite{Geyer}. This means that at
least one of the observables {\em does\,} cease to be relevant {\em
after\,} the passage of the system through the EP interface.

In our present paper we intend to contemplate an intermediate, third
possible evolution scenario in which the two sub-Hamiltonians remain
different but still matched at the critical EP instant
$t=t^{(EP)}=0$,
 \be
 \lim_{t \to 0^-}H^{(-)}_{[III]}(t) =
 \lim_{t \to 0^+}H^{(+)}_{[III]}(t)\,.
 \label{matchi}
 \ee
Some of the basic features of such a dynamical setting (to be called
the quantum phase transition of the third kind) will be illustrated
here by a direct or indirect identification of the Hamiltonians with
one of the above-mentioned toy models $H^{(N)}_{(BH/AO)}$.

\subsection{The process of degeneracy (BH example, $t<0$)}

In Ref.~\cite{Uwe} the ${\cal PT}-$symmetric BH Hamiltonian was
presented in the form in which the particle interaction was
considered small so that the spectrum remains tractable by means of
the standard mathematical power-law expansion techniques. The same
model will also prove useful in our present study. For the sake of
simplicity we will only pay attention to the unperturbed,
one-parametric version of the model. According to Ref.~\cite{Uwe}
the $N-$th-sector Hamiltonian may be then given the following
one-parametric $N$ by $N$ complex-symmetric matrix form
  \be
 H^{(N)}_{(BH)}(z)
 =\left [\begin {array}{cccccc}
  -{\rm i}(N-1)\,z &g_1&&&&\\
 g_1& -{\rm i}(N-3)\,z&\ddots&&&\\
 &g_{2}&\ddots&g_3&&
 \\
 &&\ddots&{\rm i}(N-5)\,z&g_{2}&
 \\
 &&&g_{2}&{\rm i}(N-3)\,z&g_{1}\\
 &&&&g_{1}&{\rm i}(N-1)\,z
 \end {array}\right ]\,
 \label{chamm}
 \ee
with arbitrary $N$ and with couplings
 $
 g_n=\sqrt{ (N-n)\,n}
 $
and a real parameter $z\in (-1,1)$.

In this realization it is easy to prove the existence of the two
dynamically accessible exceptional-point singularities at  $z=\pm
1$. In the context of physics these boundaries of the
unitarity-guaranteeing interval are known to be related to the
experimentally highly relevant process of the Bose-Einstein
condensation \cite{Uwe}. In the language of mathematics these EP
{\it alias\,} non-Hermitian-degeneracy singularities themselves may
be also characterized, whenever needed, as the exceptional points of
order $N$ (EPN).

One of the most interesting features of model (\ref{chamm}) may be
seen in the accessibility of the EP boundary at which the
Hamiltonian itself already becomes unphysical. The EP limit of $H$
can be characterized by the loss of its diagonalizability. In
particular, in the BH case we have $\lim_{z \to
1^-}H^{(N)}_{(BH)}(z)=H^{(N)}_{(BH)}(1)$ or $\lim_{z \to
-1^+}H^{(N)}_{(BH)}(z)=H^{(N)}_{(BH)}(-1)$. The energy spectrum
becomes fully degenerate in both cases, $E_n \to \eta$ at all $n$.
Simultaneously, the total $N-$fold degeneracy also involves the
eigenvectors \cite{Kato}. The Hamiltonian ceases to be
diagonalizable so that also the conventional time-independent
Schr\"{o}dinger equation ceases to be solvable. In practice, it is
usually replaced by its alternative
 \be
 H^{(N)}_{(BH)}(1) Q^{(N)}_{(BH)}
 = Q^{(N)}_{(BH)}\,J^{(N)}(\eta)
 \,.
 \label{Crealt}
 \ee
In such a decomposition of the EP-related non-diagonalizable
Hamiltonian one most often employs the Jordan-block matrix factor
 \be
 J^{(N)}(\eta)=\left [\begin {array}{ccccc}
    \eta&1&0&\ldots&0
 \\{}0&\eta&1&\ddots&\vdots
 \\{}0&0&\eta&\ddots&0
 \\{}\vdots&\ddots&\ddots&\ddots&1
 \\{}0&\ldots&0&0&\eta
 \end {array}\right ]\,
 \label{hisset}
 \ee
representing one of the most popular ``canonical'' forms of the
degenerate $N$ by $N$ Hamiltonian.

In our present paper, a decisive progress in understanding of the
possible role of EP-related and apparently unphysical generalized
Schr\"{o}dinger equations will be achieved. In essence, the BH
``transition matrix'' $Q^{(N)}_{(BH)}$ will be perceived as a formal
EP analogue of the conventional pseudo-unitary matrices which would
diagonalize the $z-$dependent Hamiltonian $H^{(N)}_{(BH)}(z)$ in the
non-degenerate regime with $|z|<1$. We will also show that the $z=1$
Schr\"{o}dinger-resembling decomposition (\ref{Crealt}) and its
various non-BH alternatives might be perceived as describing
instantaneous EP interfaces between the phases of a temporarily
fragile but still globally robust, unitarily evolving quantum
system.

\subsection{The process of unfolding (AO example, $t>0$)}


Our choice of the AO model of Ref.~\cite{maximal} was motivated not
only by its immediate phenomenological appeal (see, e.g., the
details in \cite{catast,catasta,catastb}) but also by its most
elementary real-matrix nature simplifying its non-numerical
tractability at {\em any\,} finite matrix dimension $N$
\cite{tridiagonal,tridiagonala,tridiagonalb}. This merit contributed
to the present project. We will see that the construction of the
evolution pattern may remain algebraic and may preserve the
unitarity {\em both\,} before and after the critical time $t=0$. For
our present purposes such a technical feature of the model will be
important. The main reason is that the matching of the two
sub-Hamiltonians $H^{(\pm)}(t)$ in Eq.~(\ref{matchi}) would be an
ill-conditioned task in a generic, purely numerical setting. In a
way indicated in \cite{GBH2}, the brute-force computer-based
construction would require the use of a variable-length arithmetics.

Our choice of at least partially solvable models facilitated,
therefore, the task of matching the systems at their $t=0$ EP
interface. Moreover, not only the AO model of Ref.~\cite{maximal}
but also its one-parametric BH partner of Ref.~\cite{Uwe} proved
tractable at any dimension $N$. Thus, recalling the real and
asymmetric $\lambda-$dependent AO-matrix Hamilotonian
  \be
 H^{(N)}_{(AO)}(\lambda)
 =\left [\begin {array}{cccccc}
  -(N-1)&g_1(\lambda)&&&&\\
 -g_1(\lambda)& -(N-3)&g_{2}(\lambda)&&&\\
 &-g_{2}(\lambda)&\ddots&\ddots&&
 \\
 &&\ddots&N-5&g_{2}(\lambda)&
 \\
 &&&-g_{2}(\lambda)&N-3&g_{1}(\lambda)\\
 &&&&-g_{1}(\lambda)&N-1
 \end {array}\right ]\,
 \label{hamm}
 \ee
one may treat the positive and not too large real parameters
$\lambda>0$ as entering the matrix elements via the following
formula,
 \ben
 g_n(\lambda)=\pm \sqrt{ (N-n)\,n\,[1-\gamma_n(\lambda)]}\,,
 \ \ \ \ N=2K\ \ {\rm or}\ \ N=2K+1\,,
 \een
 \ben
 \gamma_n(\lambda)=\lambda+\lambda^2+\ldots +
 \lambda^{K-1}+ G_n\lambda^K\,,
 \ \ \ \ \ n = 1, 2, \ldots, K
 \,.
 \een
We will choose here all of the optional constants equal to zero,
$G_n=0$. The complete EPN degeneracy may be then reached in the
$\lambda \to 0$ limit~\cite{tridiagonal,tridiagonalb}.


\section{$BH \leftrightarrow AO$ transmutations
\label{sushimi}}

\subsection{Matching of Hamiltonians}

The key idea of the construction of the phase transitions of the
third kind can be now formulated as follows. Having in mind the
EP-related decomposition (\ref{Crealt}) rewritten, at $\eta=0$, in
the following equivalent form
 \be
  J^{(N)}(0)=
 \left [Q^{(N)}_{(BH)}\right ]^{-1}\,
 H^{(N)}_{(BH)}(1)\, Q^{(N)}_{(BH)}
  \,
 \label{reCrealt}
 \ee
we have to search for an analogous formula after the replacement $BH
\to AO$ of the dynamics-representing subscripts. The recipe is
straightforward. From the limit
 \be
 H^{(N)}_{(AO)}(0)\, Q^{(N)}_{(AO)}
 = Q^{(N)}_{(AO)}\,J^{(N)}(\eta)
 \,
 \label{Rrealt}
 \ee
of Schr\"{o}dinger equation we obtain $\eta=0$ and formula
 \be
  J^{(N)}(0)=
 \left [Q^{(N)}_{(AO)}\right ]^{-1}\,
 H^{(N)}_{(AO)}(0)\, Q^{(N)}_{(AO)}
  \,.
 \label{reRrealt}
 \ee
A comparison of Eqs.~(\ref{reCrealt}) and (\ref{reRrealt}) yields,
ultimately,
 \be
   \left [Q^{(N)}_{(BH)}\right ]^{-1}\,
 H^{(N)}_{(BH)}(1)\, Q^{(N)}_{(BH)}
 =
 \left [Q^{(N)}_{(AO)}\right ]^{-1}\,
 H^{(N)}_{(AO)}(0)\, Q^{(N)}_{(AO)}
  \,.
 \label{reCRrealt}
 \ee
This result offers a key to our forthcoming constructions of the $BH
\leftrightarrow AO$ models of the quantum phase transition of the
third kind.
Relation (\ref{reCRrealt}) itself may be, indeed, read as a sample of the
Jordan-block-mediated $BH-AO$ matching (\ref{matchi}) at $t=0$.

The $BH \leftrightarrow AO$ matching formula (\ref{reCRrealt})
requires the knowledge of transition matrices $Q^{(N)}_{(BH/AO)}$.
The construction of these matrices proceeds via the solution of
equations (\ref{Crealt}) and (\ref{Rrealt}). Let us now temporarily
skip this procedure and let us postpone the construction to
sections~\ref{jofob} and~\ref{jofo}. In between, we intend to point
out that besides Eq.~(\ref{reCRrealt}) there also exist a few other,
not necessarily Jordan-block-mediated $BH-AO$ matchings of
Hamiltonians (\ref{schem}) at the EP instant $t=0$.

\subsection{Three alternative patterns of passage\label{mare}}

The numerical assignment of the canonical Jordan block
(\ref{hisset}) to a given non-diagonalizable Hamiltonian is a
nontrivial task in general \cite{GBH2}. This in fact motivated our
search for its non-numerical samples. We succeeded. We will show
below that Eqs.~(\ref{Crealt}) and (\ref{Rrealt}) are solvable in
closed form. The importance of such a result (i.e., of the
availability of the closed-form transition matrices) lies in the
possibility of a constructive illustration of some of the key
properties of the phase transitions of the third kind. This will be
provided here in a sufficiently ``realistic'' setting, i.e., in the
matrix model(s) in which the dimension $N$ may be arbitrary.

In BH model we choose the extreme $z=z^{(EP)}=1$. Recalling relation
(\ref{reCrealt}) we wrote down the non-canonical, non-Jordan but
${\cal PT}-$symmetric version
 \be
 H^{(N)}_{(BH)}(1)
 = Q^{(N)}_{(BH)}\,J^{(N)}(0)\,
 \left [Q^{(N)}_{(BH)}
 \right ]^{-1}
 \,
 \label{uCrealt}
 \ee
of the EP Hamiltonian. In the same manner we obtained
 \be
 H^{(N)}_{(AO)}(0)
 = Q^{(N)}_{(AO)}\,J^{(N)}(0)\,
 \left [Q^{(N)}_{(AO)}
 \right ]^{-1}
 \,.
 \label{uRrealt}
 \ee
Without the knowledge of transition matrices (which has to be
supplied later) this is a purely formal result. Still, its
consequences are nontrivial and remarkable. First of all, we may
apply the same transformation off the EP extreme. In this way the
transition matrices start playing the role of an optimal unperturbed
basis (see \cite{corridors} for a few further,
perturbation-theory-related comments). In this basis our
diagonalizable BH-related modified Hamiltonian may be given the form
 \be
 {L}^{(N)}_{(BH)}(z)
 =
 \left [Q^{(N)}_{(BH)}
 \right ]^{-1}
 \,H^{(N)}_{(BH)}(z)\,Q^{(N)}_{(BH)}
 \,,\ \ \ z < 1\,.
 \label{AuCrealt}
 \ee
Next, we introduce the fully analogous AO-related Hamiltonian
 \be
 {L}^{(N)}_{(AO)}(\lambda)
 = \left [Q^{(N)}_{(AO)}
 \right ]^{-1}\,H^{(N)}_{(AO)}(\lambda)\,Q^{(N)}_{(AO)}
 \,,\ \ \ \lambda > 0\,.
 \label{AuRrealt}
 \ee
Thirdly, we will be able to evaluate, in closed form, also the
``real to complex'' product matrices
 \be
 S^{(N)}_{(RC)}= Q^{(N)}_{(AO)}\,\left [Q^{(N)}_{(BH)}
 \right ]^{-1}\,
 \ee
(cf. section \ref{sec6} below). This will enable us to define the
third pair of Hamiltonian matrices
 \be
 K^{(N)}_{({BH})}(z)
 =
 S^{(N)}_{(RC)}\,H^{(N)}_{(BH)}(z)\,
 \left [S^{(N)}_{(RC)}
 \right ]^{-1},
 \label{SAuCrealt}
 \ee
 \be
 K^{(N)}_{({AO})}(\lambda)
 =
 \left [S^{(N)}_{(RC)}
 \right ]^{-1}\, H^{(N)}_{(AO)}(\lambda)\,S^{(N)}_{(RC)}
 \,.
 \label{SAuRrealt}
 \ee
With such a set of physics-controlling toy-model matrices we are now
prepared to satisfy the matching condition (\ref{matchi}).

\begin{thm}
At any matrix dimension $N<\infty$ the matching (\ref{matchi}) of
${\cal PT}-$symmetric BH model (\ref{chamm}) with ${\cal
PT}-$symmetric AO model (\ref{hamm}) admits six realizations of a
unitary quantum phase transition of the third kind. The eligible
Hamiltonians are listed in Table~\ref{dowe}.
\end{thm}

\begin{table}[h]
\caption{Six unitary-evolution realizations of the quantum
Hamiltonian of Eq.~(\ref{schem})}
 \label{dowe} \vspace{.1cm}
\centering \label{tataj}
\begin{tabular}{||l||c|c|c||}
\hline
  \hline \vspace{-0.4cm}
  &&&\\
   \hspace{.31cm} {\rm } & $H^{(-)}(t)$ & {\rm EP}
     &$H^{(+)}(t)$\\
   \hspace{.31cm} {\rm phase transition} &  &
        & \\
 \hline
      \hline
   {\rm  \hspace{.6cm} {\rm 1.  {\rm  BH $\to $ AO  }}}
   & $t=z-1 <0$\,
  &$t=0$ &$t=\lambda >0$\\
 \hline
 \hline \vspace{-0.4cm}
  &&&\\
  {\rm BH to {AO}-like} &$H_{(BH)}^{(N)}(z)$
  & $H_{(BH)}^{(N)}(0)$&$K_{({AO})}^{(N)}(\lambda)$\\ \vspace{-0.5cm}
  &&&\\
 \hline \vspace{-0.4cm}
  &&&\\
    {\rm Jordan-block match} &${L}_{(BH)}^{(N)}(z)$
  &$J^{(N)}(0)$ &${L}_{(AO)}^{(N)}(\lambda)$ \\ \vspace{-0.5cm}
  &&&\\
 \hline \vspace{-0.4cm}
  &&&\\
  {\rm {BH}-like to AO } & $K_{({BH})}^{(N)}(z)$
  & $H_{(AO)}^{(N)}(0)$ & $H_{(AO)}^{(N)}(\lambda)$\\
      \hline
 \hline
 {\rm  \hspace{.6cm} {\rm 2.  {\rm  AO $\to $ BH  }}}
   & $t=-\lambda <0$\,
  &$t=0$ &$t=1-z >0$\\
 \hline
 \hline \vspace{-0.4cm}
  &&&\\
  {\rm AO to {BH}-like} &$H_{(AO)}^{(N)}({\lambda})$
  & $H_{(AO)}^{(N)}(0)$&$K_{({BH})}^{(N)}(z)$\\ \vspace{-0.5cm}
  &&&\\
 \hline \vspace{-0.4cm}
  &&&\\
  {\rm   Jordan-block match} &${L}_{(AO)}^{(N)}({\lambda})$
  &$J^{(N)}(0)$ &${L}_{(BH)}^{(N)}(z)$\\ \vspace{-0.5cm}
  &&&\\
 \hline \vspace{-0.4cm}
  &&&\\
  {\rm {AO}-like to BH } & $K_{({AO})}^{(N)}({\lambda})$
  & $H_{(BH)}^{(N)}(0)$ & $H_{(BH)}^{(N)}(z)$
  \\
 \hline \hline
\end{tabular}
\end{table}

\begin{proof}
Three independent realizations of the degenerate Hamiltonian are, at
the instant of matching $t=0$, at our disposal. Off the EP regime,
the choice of the dependence of the parameters on time is fully at
our disposal. In principle, it can be of two types. In the first one
we may mimic the BH $\to $ AO evolution, and we may work with the
growing $z$ (say, with $z=1+t$ at negative $t<0$) and with the
growing $\lambda$ (say, with $\lambda=t$ at $t>0$). The inverse AO
$\to $ BH processes are obtained when we modify the scheme and when
we postulate the decrease of $\lambda=-t$ (at $t<0$) and of $z=1-t$
(at $t>0$).
\end{proof}

\section{Non-numerical construction (BH case)\label{jofob}}


The explicit Jordan-form reduction (\ref{reCrealt}) of our complex
BH Hamiltonian $H^{(N)}_{(BH)}(1)$ required the knowledge of the
$N^2$ complex matrix elements of transition matrix $Q^{(N)}_{(BH)}$.
The transition matrix itself is determined by the set of the $N^2$
linear algebraic equations (\ref{Crealt}). The solution
$Q^{(N)}_{(BH)}$ of such a generalized Schr\"{o}dinger equation is,
up to an arbitrary overall multiplication factor, unique. We
revealed that such a solution may be obtained in closed form. Let us
now describe the details.

\subsection{Transition matrices $Q^{(N)}_{(BH)}$
at small $N$\label{sec4Rbv}}

At the smallest matrix dimensions $N$ it makes sense to use a
suitable symbolic-manipulation solver (viz., MAPLE \cite{Maple} in
our case). Using a completely routine procedure we were able to list
the most elementary BH Hamiltonians and to evaluate, quickly, the
related transition matrices,
  $$
 H^{(2)}_{(BH)}=\left[ \begin {array}{cc} -i&1\\\noalign{\medskip}1&i\end {array} \right]
 \,,\ \ \ \ \
 Q^{(2)}_{(BH)}=\left[ \begin {array}{cc} -i&1\\\noalign{\medskip}1&0\end {array}
 \right]\,,
 $$
 $$
 H^{(3)}_{(BH)}=\left[ \begin {array}{ccc}
 -2\,i&\sqrt {2}&0\\\noalign{\medskip}\sqrt {2}&0&\sqrt {2}
 \\\noalign{\medskip}0&\sqrt {2}&2\,i\end {array}
 \right]
 \,,\ \ \ \ \
 Q^{(3)}_{(BH)}= \left[ \begin {array}{ccc} -2&-2\,i&1\\\noalign{\medskip}-2\,i\sqrt {
2}&\sqrt {2}&0\\\noalign{\medskip}2&0&0\end {array} \right]\,,\
\ldots\ .
 $$
At the higher matrix dimensions $N$ one may encounter serious
technical difficulties even for the BH-related complex-symmetric
matrices. The reasons were explained in Ref.~\cite{GBH2}. In
essence, these difficulties may only partly be attributed to the
above-mentioned ill-conditioned nature of Eq.~(\ref{Crealt}).
Fortunately, for the BH models of Ref.~\cite{Uwe} such an obstacle
proved softened by the specific symmetries of the model. The
existence of these symmetries explains why the construction remained
straightforward even for the $N=6$ BH Hamiltonian
 $$
 H^{(6)}_{(BH)}=\left[ \begin {array}{cccccc} -5\,i&\sqrt {5}&0&0&0&0
 \\\noalign{\medskip}\sqrt {5}&-3\,i&2\,\sqrt {2}&0&0&0
 \\\noalign{\medskip}0&2\,\sqrt {2}&-i&3&0&0\\\noalign{\medskip}0&0&3&i
 &2\,\sqrt {2}&0\\\noalign{\medskip}0&0&0&2\,\sqrt {2}&3\,i&\sqrt {5}
 \\\noalign{\medskip}0&0&0&0&\sqrt {5}&5\,i\end {array} \right]\,.
 $$
This operator was still assigned the transition matrix in closed
form,
 \be
 Q^{(6)}_{(BH)}=\left[ \begin {array}{cccccc} -120\,i&120&60\,i&-20&-5\,i&1
 \\\noalign{\medskip}120\,\sqrt {5}&96\,i\sqrt {5}&-36\,\sqrt {5}&-8\,i
 \sqrt {5}&\sqrt {5}&0\\\noalign{\medskip}120\,i\sqrt {5}\sqrt
 {2}&-72 \,\sqrt {5}\sqrt {2}&-18\,i\sqrt {5}\sqrt {2}&2\,\sqrt
 {5}\sqrt {2}&0&0
 \\\noalign{\medskip}-120\,\sqrt {5}\sqrt {2}&-48\,i\sqrt {5}\sqrt {2}&
 6\,\sqrt {5}\sqrt {2}&0&0&0\\\noalign{\medskip}-120\,i\sqrt {5}&24\,
 \sqrt {5}&0&0&0&0\\\noalign{\medskip}120&0&0&0&0&0\end {array}
 \right]\,.
 \label{sestac}
 \ee
The matrix elements of the latter matrix proved to have the form
which seems to admit extrapolation.

\subsection{Transition matrices $Q^{(N)}_{(BH)}$
at all $N$ \label{sec4Cb}}


A few trial and error experiments with an $N>6$  extrapolation of
formula (\ref{sestac}) led finally to the promising ansatz, the
correctness of which was then easily verified, by insertion, at a
few higher matrix dimensions $N$. Ultimately, the use of
mathematical induction confirmed that the closed-form ansatz which
defined the transition matrices in the symbolic-manipulation test is
correct.

\begin{lemma}
At all $N$, transition matrices $Q^{(N)}_{(BH)}$ may be factorized,
 $$
 Q^{(N)}_{(BH)}=
 D^{(N)}
 \times
 P^{(N)}
 \times
 G^{(N)}\,.
 $$
The formula contains the two complex diagonal-matrix factors with
elements
 $$
 D^{(N)}_{n,n}={\rm i}^n\,\times\,\sqrt{\left ( \ba N-1\\N-n
 \ea
 \right )}\,,\ \ \ n = 0, 1, \ldots, N-1
 $$
and
 $$
 G^{(N)}_{n,n}=(-i)^{N-n-1}\,\times\,(N-1-n)!\,,\ \ \ n = 0, 1, \ldots, N-1\,
 $$
together with the Pascal-triangle matrix
 \be
 P^{(N)}_{m,q}=\left ( \ba N-1-m\\q
 \ea
 \right )\,,\ \ \ \ m,q=0,1,\ldots, N-1
 \,.
 \label{pascal}
 \ee
\end{lemma}
 \begin{proof}
Direct computations yield
 \be
 P^{(4)}=\left[ \begin {array}{cccc} 1&3&3&1
 \\\noalign{\medskip}1&2&1&0
 \\\noalign{\medskip}1&1&0&0
 \\\noalign{\medskip}1&0&0&0\end {array} \right]\,, \ \ \
 P^{(5)}=\left[ \begin {array}{ccccc} 1&4&6&4&1
 \\\noalign{\medskip}1&3&3&1&0
 \\\noalign{\medskip}1&2&1&0&0
 \\\noalign{\medskip}1&1&0&0&0
 \\\noalign{\medskip}1&0&0&0&0
 \end {array} \right]
 \label{pasca}
 \ee
etc. The general Pascal-triangle form of matrices $P^{(N)}$ is
easily deduced and, subsequently, proved by mathematical induction.
Such a confirmation of the validity of relation (\ref{Crealt}) only
requires a reduction of the latter equation to entirely elementary
combinatorial identities.
\end{proof}

\section{Non-numerical construction (AO case)
\label{jofo}}


The real asymmetric toy-model Hamiltonian $H^{(N)}_{(AO)}(\lambda)$
is maximally elementary but, in comparison with its BH partner, it
is more formal and motivated by mathematics rather than physics. It
did not attract enough attention of experimentalists yet, probably
due to the lack of a recipe of its simulation in the laboratory. At
the EP singularity its most obvious contact with experiments might
be deduced from its complex-symmetric-matrix rearrangement
(\ref{reCRrealt}). It can be given the explicit-elimination form
 \be
 H^{(N)}_{(AO)}= S^{(N)}_{RC}\,H^{(N)}_{(BH)}\,
 \left [S^{(N)}_{RC}\right ]^{-1}\,,
 \ \ \ \ \
 S^{(N)}_{RC} =
 Q^{(N)}_{(AO)}
 \left [Q^{(N)}_{(BH)}\right ]^{-1}\,.
 \label{fifn}
 \ee
In the unitary-evolution-compatible vicinity of this singularity the
$BH - AO$ correspondence can still be re-established using the
suitable forms of perturbation theory \cite{corridors}.

\subsection{Transition matrices $Q^{(N)}_{(AO)}$
at small $N$\label{sec4Rb}}

At the small dimensions $N$ of Hamiltonians (\ref{hamm}), i.e., for
Hamiltonians
 $$
 H^{(2)}_{(AO)}(0)=\left[ \begin {array}{cc} -1&1\\\noalign{\medskip}-1&1\end {array}
 \right]\,,\ \ \ \
  H^{(3)}_{(AO)}(0)=\left[ \begin {array}{ccc} -2&\sqrt {2}&0\\\noalign{\medskip}
 -\sqrt {2}&0&\sqrt {2}\\\noalign{\medskip}0&-\sqrt {2}&2\end {array}
 \right]\,,
 $$
etc, the process of solving Eq.~(\ref{Rrealt}) is entirely routine.
Up to $N=8$ we verified that for our particular model it yields the
sequence of fully non-numerical results
 $$
 Q^{(2)}_{(AO)}=\left[ \begin {array}{cc} -1&1\\\noalign{\medskip}-1&0\end {array}
 \right]
\,,\ \ \
 Q^{(3)}_{(AO)}= \left[ \begin {array}{ccc} 2&-2&1\\\noalign{\medskip}2\,\sqrt
 {2}&-
\sqrt {2}&0\\\noalign{\medskip}2&0&0\end {array} \right]\,,
 $$
 \be
 Q^{(4)}_{(AO)}=\left[ \begin {array}{cccc} -6&6&-3&1
 \\\noalign{\medskip}-6\,\sqrt {3}&4\,\sqrt {3}&-\sqrt {3}&0
 \\\noalign{\medskip}-6\,\sqrt {3}&2\,\sqrt
{3}&0&0\\\noalign{\medskip}-6&0&0&0\end {array} \right]\,, \ \ \
 Q^{(5)}_{(AO)}=\left[ \begin {array}{ccccc} 24&-24&12&-4&1
 \\\noalign{\medskip}48&-36&12&-2&0
 \\\noalign{\medskip}24\,\sqrt {6}&-12\,\sqrt {6}&2\,\sqrt {6}&0
&0\\\noalign{\medskip}48&-12&0&0&0\\\noalign{\medskip}24&0&0&0&0
\end {array} \right]
\label{arrr}
 \ee
etc.

\subsection{Transition matrices $Q^{(N)}_{(AO)}$
at all $N$ \label{sec4C}}


The inspection of sequence (\ref{arrr}) offers an insight in the
general case.

\begin{lemma}
The transition matrices may be factorized into three-term products
 $$
 Q^{(N)}_{(AO)}=
 C^{(N)}
 \times
 P^{(N)}
 \times
 F^{(N)}
 $$
containing the same Pascal-triangle matrix $P^{(N)}$ as above [cf.
Eqs.~(\ref{pasca}) and~(\ref{pascal})]. The pre- and post-factors
are diagonal matrices
 $$
 C^{(N)}_{n,n}=\sqrt{\left ( \ba N-1\\N-n
 \ea
 \right )}\,,\ \ \ n = 0, 1, \ldots, N-1
 $$
and
 $$
 F^{(N)}_{n,n}=(-1)^{N-n-1}(N-1-n)!\,,\ \ \ n = 0, 1, \ldots, N-1\,.
 $$
\end{lemma}
The elementary proof using mathematical induction is left to the
readers.

\section{Direct $AO - BH$ correspondence\label{sec6}}


The original aim of our present paper was an explicit construction
of the $N-$dependent mappings $BH \leftrightarrow AO$  using the
Jordan-block-based intermediate representation of the process at
$\lambda=0$ and $z=1$. Along these lines we obtained the two unitary
Jordan-block-mediated EP-passing evolution models based on {\it ad
hoc\,} redefinitions of the respective BH and AO Hamiltonians,
 \be
 H^{(N)}_{(AO)}(0) \ \to \ L^{(N)}_{(AO)}(0) \ \to \
 J^{(N)}(0) \ \leftarrow \
 L^{(N)}_{(BH)}(1) \ \leftarrow \
 H^{(N)}_{(BH)}(1)\,.
 \label{3step}
 \ee
The redefined matrices $L^{(N)}_{(BH/AO)}(z/\lambda)$ were obtained
using the exact transition matrices $Q^{(N)}_{(BH/AO)}$ (cf. the
respective Eqs.~(\ref{AuCrealt}) and (\ref{AuRrealt}) and/or
sections \ref{jofob} and \ref{jofo} and/or two lines in Table
\ref{tataj}).

In the remaining four evolution scenarios of Table \ref{tataj} we
assumed the knowledge of the other two redefined Hamiltonians
$K^{(N)}_{(BH/AO)}(z/\lambda)$. Let us now complete the picture by
adding the non-numerical version of the underlying necessary
redefinitions of the Hamiltonian.

\subsection{Transition-matrix products $S^{(N)}_{RC}$
at small $N$\label{sec4Rbv}}

Six years ago the construction of the direct $AO - BH$
correspondence was still perceived as a numerical
task~\cite{comment}. We reopened the problem recently, and we are
now able to obtain the transition-matrix-products $S^{(N)}_{RC}$ of
Eq.~(\ref{fifn}) in a fully non-numerical form.

The realization of such a project will be split in two subprojects.
Firstly, at the small matrix dimensions $N$ the construction of the
real $\to$ complex transformation remains non-numerical. The
simplicity and the straightforward invertibility of transition
matrices $Q^{(N)}_{(AO)}$ and $Q^{(N)}_{(BH)}$ enabled us to
evaluate, by brute force, the sequence of solutions up to $N=5$,
 $$
 S^{(2)}_{RC} =\left[ \begin {array}{cc} 1&-1+i
 \\\noalign{\medskip}0&-1\end {array} \right]\,,\ \ \ \ \
 S^{(3)}_{RC} =\left[ \begin {array}{ccc}
 1&-\sqrt {2}+i\sqrt {2}&-2\,i\\\noalign{\medskip}0&-1&
 \sqrt {2}-i\sqrt {2}\\\noalign{\medskip}0&0&1
\end {array} \right]\,,
 $$
 $$
 S^{(4)}_{RC} =\left[ \begin {array}{cccc} 1&-\sqrt {3}+i\sqrt {3}&-2\,i\sqrt {3}&2+2\,i
 \\\noalign{\medskip}0&-1&2-2\,i&2\,i\sqrt {3}\\\noalign{\medskip}0
&0&1&-\sqrt {3}+i\sqrt {3}\\\noalign{\medskip}0&0&0&-1\end {array}
 \right]\,,
 $$
 %
%
%
%
 $$
 S^{(5)}_{RC}= \left[ \begin {array}{ccccc} 1&-2+2\,i&-2\,i\sqrt {6}&4+4\,i&-4
\\\noalign{\medskip}0&-1&\sqrt {6}-i\sqrt {6}&6\,i&-4-4\,i
\\\noalign{\medskip}0&0&1&-\sqrt {6}+i\sqrt {6}&-2\,i\sqrt {6}
\\\noalign{\medskip}0&0&0&-1&2-2\,i\\\noalign{\medskip}0&0&0&0&1
\end {array} \right]\,.
 $$
These results became a starting point of extrapolations.

\subsection{Transition-matrix products $S^{(N)}_{RC}$
at arbitrary $N$\label{sec4Cd}}

The apparently complicated structure of complex matrices
$S^{(N)}_{RC}$ at $N \leq 5$ proved thoroughly simplified by their
decomposition into products
 \be
 S^{(N)}_{RC}=
 B^{(N)}
 \times
 R^{(N)}_{RC}
 \times
 A^{(N)}
 \label{hulan}
 \ee
of a diagonal matrix $B^{(N)}$ times a strictly real matrix of
square roots of integers times another diagonal matrix $A^{(N)}$. In
terms of an auxiliary complex constant $\beta=-1+{\rm i}$ we
obtained
 \be
 B^{(N)}_{n,n}=(-\beta)^{-n}\,
 \,,\ \ \ \ \ \
 A^{(N)}_{n,n}=\beta^n\,,\ \ \ \ \ \
   n = 0, 1, \ldots, N-1\,.
 \label{paskr}
 \ee
At the smallest matrix dimensions $N$ the test calculations yielded
 $$
 R^{(2)}_{RC}= \left[ \begin {array}{cc} 1&1
 \\\noalign{\medskip}0&1\end {array} \right]\,,\ \ \ \ \
%
%
 R^{(3)}_{RC}=\left[ \begin {array}{ccc} 1&\sqrt {2}&1
 \\\noalign{\medskip}0&1&\sqrt {2}
 \\\noalign{\medskip}0&0&1\end {array} \right]\,,
\ \ \ \ \
%
%
 R^{(4)}_{RC}=\left[ \begin {array}{cccc} 1&\sqrt {3}&\sqrt {3}&
 1\\\noalign{\medskip}0&1&2&\sqrt {3}\\\noalign{\medskip}0&0&1&\sqrt {3}
\\\noalign{\medskip}0&0&0&1\end {array} \right]\,,
 $$
 $$
 R^{(5)}_{RC}= \left[ \begin {array}{ccccc} 1&2&\sqrt {6}&2&1
 \\\noalign{\medskip}0&1
&\sqrt {6}&3&2\\\noalign{\medskip}0&0&1&\sqrt {6}&\sqrt {6}
\\\noalign{\medskip}0&0&0&1&2\\\noalign{\medskip}0&0&0&0&1\end {array}
 \right]
 \,,
\ \ \ \ \
 R^{(6)}_{RC}= \left[ \begin {array}{cccccc} 1&\sqrt {5}&\sqrt {5}\sqrt {2}&\sqrt {5
}\sqrt {2}&\sqrt {5}&1\\\noalign{\medskip}0&1&2\,\sqrt {2}&3\,\sqrt
{2 }&4&\sqrt {5}\\\noalign{\medskip}0&0&1&3&3\,\sqrt {2}&\sqrt
{5}\sqrt { 2}\\\noalign{\medskip}0&0&0&1&2\,\sqrt {2}&\sqrt {5}\sqrt
{2}
\\\noalign{\medskip}0&0&0&0&1&\sqrt {5}\\\noalign{\medskip}0&0&0&0&0&1
\end {array} \right]
 $$
etc. The $N-$dependence of these results did not seem to exhibit any
obvious regularities. In our search for these regularities we really
had to evaluate several further elements of the matrix sequence
$R^{(N)}_{RC}$ with $N \geq 6$ in order to be able to formulate some
productive extrapolation hypotheses.

We succeeded. Unfortunately, it would be difficult to display the
larger, $N \geq 6$ matrices $R^{(N)}_{RC}$ in print. Still, we have
to emphasize that their role in the search for the correct
extrapolation pattern was essential. Due to the still sufficiently
elementary combinatorial-number nature of the matrix elements in
question, the successful tentative ansatzs were comparatively easily
shown to be correct. With the rigorous proof provided by
mathematical induction, our final result may be formulated as
follows.

\begin{thm}
At any finite dimension $N$, the one-to-one correspondence
(\ref{fifn}) between the real asymmetric Hamiltonians
$H^{(N)}_{(AO)}(0)$ of Eq.(~\ref{hamm}) and the complex symmetric
Hamiltonians $H^{(N)}_{(BH)}(1)$ of Eq.(~\ref{chamm}) is mediated by
the upper triangular matrices $S^{(N)}_{RC}$ defined as products
(\ref{hulan}) of the complex diagonal pre- and post-factor
(\ref{paskr}) with the upper triangular real matrix of square-roots
of integers, with elements
 \be
 \left [R^{(N)}_{RC}\right ]_{n,n+q}=
 \sqrt{\left (
 \ba
 n+q\\
 n
 \ea
 \right )
 \left (
 \ba
 N-1-n\\
 q
 \ea
 \right )
 }
  \label{uprka}
 \ee
where
 $
 \,
 q=0,1,\ldots,N-1-n\,$ and $
 \,
 n=0,1,\ldots,N-1$.
\end{thm}
The straightforward proof using mathematical induction is again left
to the readers.


\section{Discussion\label{summary}}

\subsection{A note on terminology\label{teos}}

A word of warning should be issued concerning the present
terminology. Firstly, our usage of the term of  ``quantum phase
transition'' was inspired by the related older papers \cite{denis}.
Indeed, this usage might interfere with the same name frequented in
conventional Hermitian-operator descriptions of certain specific
$N-$body quantum systems exhibiting spontaneous symmetry breaking to
different ground states in the $N \to \infty$ limit at zero
temperature. We believe that such a terminological ambiguity cannot
lead to any confusion in the present context.

Secondly, the name of the exceptional point is also occurring in
different physical, not necessarily too closely related
applications. Invented as a useful concept in the purely
mathematical theory of perturbations \cite{Kato}, its possible
implications and appeal broadened. Nowadays they cover many branches
of quantum as well as non-quantum phenomenology ranging from
photonics \cite{photonics} up to the effective, manifestly
non-unitary theories of open quantum systems
\cite{bosch,Berryb,Liouvillians}.

Last but not least we have to add that in our present non-numerical
constructive analysis of various non-Hermitian time-dependent
Hamiltonian operators $H(t)$ admitting the various EP-related
quantum phase transition phenomena emerging at a specific, singular
instant of time $t=t^{(EP)}=0$, many deep theoretical questions
still remain unanswered. Many of them are connected even with the
precise meaning of the innocent-looking concept of quantum
Hamiltonians. Unfortunately, it is beyond the scope of the present
brief note to provide any deeper explanation of the problem.
Interested readers have to be redirected to our recent comprehensive
review \cite{NIP}.

Just {\it pars pro toto} let us mention that the essence of the
problems with the {\em dynamical\,} non-stationarity of the
quasi-Hermitian Hamiltonian operators $H(t)$ lies in the fact that
the related conventional time-dependent Schr\"{o}dinger equation
 \be
 {\rm i}\partial_t \psi(t) = H(t)\,\psi(t)
 \label{seti}
 \ee
only describes the evolution of the underlying unitary quantum
system in question under certain not entirely trivial assumptions.
In this respect and, in particular, in relation to the currently
unresolved problem of the specification of boundaries of the
applicability domain of the adiabatic theorem(s), a warmly
recommended reading may be found in paper \cite{Miranda}.

\subsection{Quantum crossroads\label{cros}}

The limit $t \to t^{({EP})}=0$ represents a quantum catastrophe at
which the Hamiltonian $H^{(-)}(t)$ controlling the $t<0$ unitary
evolution ceases to be diagonalizable. Still, a continuation of the
evolution to positive times may make sense: the quantum system in
question simply forgets about its past and performs a phase
transition. Such a point of view remains phenomenologically
meaningful. The choice of the $t>0$ dynamics controlled by
$H^{(+)}(t)$ remains fully in our hands.

As a consequence, the dynamics prescribed by Eq.~(\ref{schem}) can
be generalized and replaced by a branched recipe. Thus, for example,
the $t>0$ future can be assumed controlled by a richer menu of
optional Hamiltonians $H^{(+)}_{(A)}(t)$, $H^{(+)}_{(B)}(t)$, etc. A
selection out of this menu could be controlled, for example, by the
respective probabilities $p_{(A)}$, $p_{(A)}$, etc. We may conclude
that in applications, unitary evolution may branch and, at the EP
indeterminacy instant $t=0$, bifurcate as follows,
 \ben
 \ba
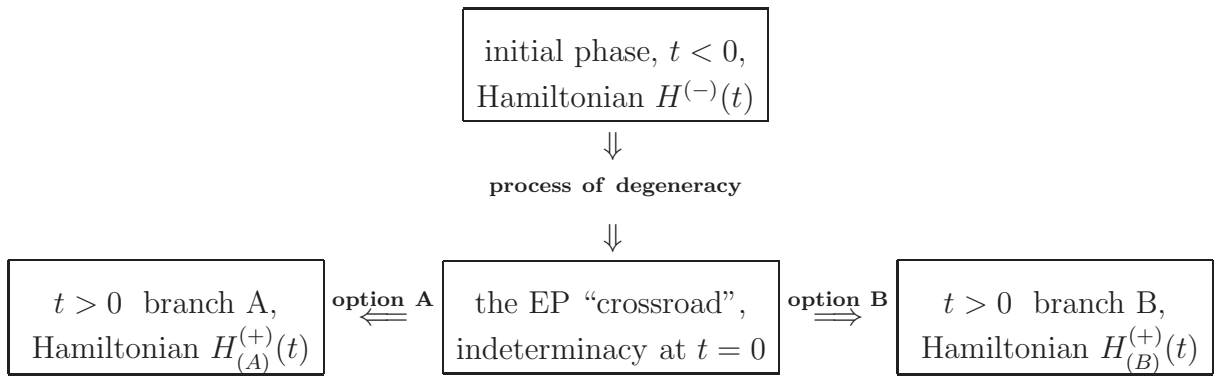
 \begin{array}{|c|}
  \hline
 \vspace{-0.35cm}\\
         {\rm initial\ phase,}\ t<0, \\
    {\rm  Hamiltonian\ } H^{(-)}(t)\\
 \hline
 \ea
\\
  \Downarrow \,\\
    \stackrel
 { \bf process \ of\ degeneracy}{\small \rm }
 \\
  \Downarrow \,\\
    \begin{array}{|c|}
 \hline
 \vspace{-0.35cm}\\
 t>0\ {\rm \ branch\ A,}\\
    {\rm \ Hamiltonian\ } H^{(+)}_{(A)}(t)\\
  \hline
 \ea
 \stackrel{ {\bf   option\ A}  }{ \Longleftarrow }
 \begin{array}{|c|}
 \hline
 \vspace{-0.35cm}\\
      {\rm the\ EP\ ``crossroad",}\\
  {\rm indeterminacy\ at}\ t=0 \\
   %
 \hline
 \ea
 \stackrel{ {\bf  option\ B}  }{ \Longrightarrow }
 \begin{array}{|c|}
 \hline
 \vspace{-0.35cm}\\
 t>0\ {\rm \ branch\ B,}\\
    {\rm \ Hamiltonian\ } H^{(+)}_{(B)}(t)\\
 \hline
 \ea \
  \\
 \ea
 \een
Thus, in principle, the future extensions of our present models
might even incorporate a multiverse-resembling branching of
evolutions at $t=0$.

\section{Summary\label{sumsummary}}

The present specific, non-branching toy-model explicit closed-form
realizations of the eligible quantum  BH $\leftrightarrow$ AO phase
transitions of the third kind are listed in Table~\ref{tataj}. As
many as six different unitary-evolution scenarios were identified
there, with the reference to their respective generator matrices in
subsequent technical sections.

The idea of the feasibility of such constructions was inspired by
the availability of multiple nontrivial $N$ by $N$ matrix
Hamiltonians $H^{(\mp)}_{[III]}(t)$ admitting the EP degeneracy in
the literature. For our purposes we choose the realizations of
$H^{(\mp)}_{[III]}(t)$ via the complex symmetric BH model of
Ref.~\cite{Uwe}, and via the real, asymmetric but ${\cal
PT}-$symmetrized $N$ by $N$ matrix version of the discrete
anharmonic oscillator of Ref.~\cite{maximal}. {\it A priori}, these
models looked particularly useful even though their explicit $t=0$
matching was considered, for several years, a difficult open problem
\cite{comment}.

A decisive encouragement and technical hints of its present solution
only came with the recent clarification of the concept of smallness
of perturbations in the non-Hermitian but hiddenly unitary quantum
systems living in a vicinity of their natural EP-related boundaries
of stability \cite{corridors}. On this background the key
technicality has finally been revealed to lie in an extreme
enhancement of the anisotropy of the geometry of the physical
Hilbert space of states \cite{lotor} near the EP extreme. Thus, our
efforts were redirected to the closed-form constructions of the
metric and, subsequently, to the discovery of the non-numerical form
of the triplet of the fundamental basis-transformation matrices
$Q^{(N)}_{(BH)}$, $Q^{(N)}_{(AO)}$ and $R^{(N)}_{(RC)}$ forming the
most relevant ingredients of our resulting ultimate insight in the
physical mechanism and mathematical aspects of the EP-related
unitary quantum phase transitions of the third kind.

\section*{Acknowledgements}

This work was not supported by grants.


\end{document}